\newif\ifconf
\conftrue 
\conffalse

  \documentclass[letterpaper,11pt]{article}

  \usepackage{typearea}
  \typearea{15}

  \usepackage[compact]{titlesec}

\usepackage{theorem,latexsym,graphicx}
\usepackage{amsmath,amssymb,enumerate}
\usepackage{xspace}

\usepackage{bm}
\usepackage{mathrsfs}
\usepackage{ifpdf}
\usepackage{shadow,shadethm,color}
\usepackage{algorithm}
\usepackage[noend]{algorithmic}
\usepackage{subfigure}
\usepackage{verbatim}
\usepackage{paralist}

\allowdisplaybreaks

\ifconf
\newcommand{\lref}[2][]{{#1~\ref{#2}}}
\else
\definecolor{Darkblue}{rgb}{0,0,0.4}
\definecolor{Brown}{cmyk}{0,0.81,1.,0.60}
\definecolor{Purple}{cmyk}{0.45,0.86,0,0}
\newcommand{\mydriver}{hypertex}
\ifpdf
 \renewcommand{\mydriver}{pdftex}
\fi
\usepackage[breaklinks,\mydriver]{hyperref}
\hypersetup{colorlinks=true,
            citebordercolor={.6 .6 .6},linkbordercolor={.6 .6 .6},%
citecolor=Darkblue,urlcolor=black,linkcolor=Darkblue,pagecolor=black}
\newcommand{\lref}[2][]{\hyperref[#2]{#1~\ref*{#2}}}
\fi


\newtheorem{theorem}{Theorem}[section]

\newtheorem{proposition}[theorem]{Proposition}

\newtheorem{lemma}[theorem]{Lemma}
\newshadetheorem{lemmashaded}[theorem]{Lemma}

\newtheorem{corollary}[theorem]{Corollary}

\numberwithin{algorithm}{section}

\newenvironment{proof}{

\noindent{\bf Proof:}}
{\hfill$\blacksquare$

}

\newcommand{\junk}[1]{}
\newcommand{\ignore}[1]{}

\newcommand{\R}[0]{{\ensuremath{\mathbb{R}}}}

\newcommand{\argmin}{\operatorname{argmin}}
\newcommand{\argmax}{\operatorname{argmax}}

\newcommand{\sse}{\subseteq}
\newcommand{\B}{{\mathbf{B}}}
\newcommand{\cE}{{\mathcal{E}}}

\newcommand{\eps}{\varepsilon}

\newcommand{\Ex}{{\rm E}}

\newcommand{\expectation}[3][0]{%
  \ifcase#1
  E( #2 \mid #3 )
  \or E \bigl( #2 \bigm\vert #3 \bigr)
  \or E \Bigl( #2 \Bigm\vert #3 \Bigr)
  \or E \biggl( #2 \biggm\vert #3 \biggr)
  \or E \Biggl( #2 \Biggm\vert #3 \Biggr)
  \else
  E \left( #2  \;\middle\vert\; #3 \right)
  \fi}

\newcounter{note}[section]

\newcommand{\qedsymb}{\hfill{\rule{2mm}{2mm}}}
\renewenvironment{proof}{\begin{trivlist} \item[\hspace{\labelsep}{\bf
\noindent Proof.\/}] }{\qedsymb\end{trivlist}}%

\newcommand{\initOneLiners}{%
    \setlength{\itemsep}{0pt}
    \setlength{\parsep }{0pt}
    \setlength{\topsep }{0pt}
}
\newenvironment{OneLiners}[1][\ensuremath{\bullet}]
    {\begin{list}
        {#1}
        {\initOneLiners}}
    {\end{list}}

\newcommand{\squishlist}{
 \begin{list}{$\bullet$}
  { \setlength{\itemsep}{0pt}
     \setlength{\parsep}{3pt}
     \setlength{\topsep}{3pt}
     \setlength{\partopsep}{0pt}
     \setlength{\leftmargin}{1.5em}
     \setlength{\labelwidth}{1em}
     \setlength{\labelsep}{0.5em} } }

\newcommand{\squishend}{
  \end{list}  }

\newcommand{\thresh}{\mathscr{T}}
\newcommand{\good}{\ensuremath{\cE}\xspace}
\newcommand{\hatr}{\widehat{\rho}}
\newcommand{\rhothresh}{\rho^c}

\begin{document}

\title{Random Rates for $0$-Extension and Low-Diameter Decompositions}
\author{
Anupam Gupta\thanks{Department of Computer Science, Carnegie Mellon
  University, Pittsburgh PA 15213, and Microsoft Research SVC, Mountain
  View, CA 94043. Research was partly supported by
    NSF awards CCF-0964474 and CCF-1016799, and by a grant from the
    CMU-Microsoft Center for Computational Thinking.}
\and
Kunal Talwar\thanks{Microsoft Research SVC, Mountain View, CA 94043.}
}
\date{}

\maketitle

\begin{abstract}
  Consider the problem of partitioning an arbitrary metric space into
  pieces of diameter at most $\Delta$, such every pair of points is
  separated with relatively low probability. We propose a rate-based
  algorithm inspired by multiplicatively-weighted Voronoi diagrams, and
  prove it has optimal trade-offs. This also gives us another algorithm
  for the $0$-extension problem.
\end{abstract}

\section{Introduction}
\label{sec:introduction}

We consider partitioning problems of the following form: given a metric
$(V,d)$, how should we decompose it into ``small'' pieces so as to cut
``few'' edges. There are many variants of this general form, and in this
note we consider two of them: \emph{terminal
  partitioning/$0$-extension} and \emph{low-diameter decompositions}.

In the \emph{low-diameter decomposition problem}, we are given a metric
$(V,d)$ and a diameter bound $\Delta$, and the goal is to (randomly)
partition the set $V$ into pieces each of diameter at most $\Delta$ so
that for any pair $x,y \in V$,
\[ \Pr[ \, x,y \text{ separated} \,] \leq \beta \cdot
\frac{d(x,y)}{\Delta}. \] It is known that $\beta = O(\log n)$ is the
possible for any $n$-point metric, and there are metrics for which no
better is possible. Such decompositions have been widely studied, e.g.,
works by Awerbuch~\cite{Awer85}, Linial and Saks~\cite{LinialSaks93},
Leighton and Rao~\cite{LR88}, Garg, Vazirani, and
Yannakakis~\cite{GVY96}, and Seymour~\cite{Sey93} studied an equivalent
deterministic version of this problem, and Bartal~\cite{Bar96},
Calinescu, Karloff, and Rabani~\cite{CKR01-zero}, Fakcharoenphol, Rao,
and Talwar~\cite{FRT03}, and Abraham, Bartal, and Neiman~\cite{ABN06}
studied randomized versions. (This is almost certainly an incomplete
list --- though some other pertinent references follow.)  Many of these
results study more nuanced parameters and give bounds that improve on
$O(\log n)$ for special cases, but we omit discussions of these for sake
of brevity.

The \emph{terminal partitioning problem} can be thought of as a
multi-scale version of low-diameter decomposition. This name is not
standard (we coin it here), but it arises in solving the
\emph{$0$-extension} problem. In terminal partitioning, instead of a
diameter bound, we are given a set $T$ of terminals, where $T \sse V$
and $|T| = k$, and we want a (random) partition $V_1, V_2, \ldots, V_k$,
such that the $i^{th}$ terminal $t_i \in V_i$, and for any $x,y \in V$,
\[ \Pr[ \, x,y, \text{ separated} \,] \leq \alpha \cdot
\frac{d(x,y)}{\min\{d(x,T), d(y,T)\}}. \] In other words, edges whose
endpoints are far away from the terminal set should be cut with smaller
probability than edges whose endpoints are close to terminals, a natural
enough requirement. Again, it is known that $O(\log k)$ is possible for
any metric~\cite{CKR01-zero}; however, this is not the best possible in this
case~\cite{FHRT03}. 

The writing of this note was prompted by two elegant recent results. The
first is a paper of Buchbinder, Naor, and Schwartz~\cite{BNS13} that studies
the multiway cut problem, which is a special case of $0$-extension. They
give a rounding based on exponential clocks. (An identical rounding was
earlier, though independently, also given by Ge et al.~\cite{GHYZ11}.)
The second is a paper of Miller, Peng, and Xu~\cite{MPX13}, who study
low-diameter decompositions and give a algorithm with $\beta = O(\log
n)$ based on exponential clocks. Their algorithm is easily
parallelizable, and it substantially improves and cleans up a previous
sub-optimal algorithm in the parallel setting due to Blelloch et
al.~\cite{BGKMPT11-tocs}.

\subsection{Our Results}

In this note we give an algorithm for the terminal partitioning problem,
which has $\alpha = O(\log k)$. This immediately gives an $O(\log k)$
approximation for the $0$-extension problem. While this ratio is not
optimal, we find the algorithm appealing due to its simplicity: for each
terminal $t \in T$, we pick a random rate $\rho_t$ from a certain
(shifted, truncated exponential) probability distribution.\footnote{The random
  variable $\rho_t \sim 1 + \text{Exp}(\ln k)$ conditioned on being at most 2; details follow in
  Section~\ref{sec:zero-extension}.}  Then for each non-terminal $v \in
V$, we assign it to the terminal
\[ \argmin_{t \in T} \left\{ \frac{d(x,t)}{\rho_t} \right\} \] breaking
ties arbitrarily. (This is very similar in spirit to the~\cite{BNS13,
  GHYZ11} geometric rounding for multiway cut simplex linear program.)

A side-effect of our algorithm for terminal partitioning is a certain
``proximity'' condition: it only assigns each vertex to ``close-by''
terminals. We show that terminal partitionings that satisfy this kind of
proximity condition also give us low-diameter decompositions, merely by
choosing an $O(\Delta)$-net of the metric as the terminal set and then
running the terminal partitioning algorithm. This immediately gives a
low-diameter decomposition with $\beta = O(\log n)$, which is best
possible. Details appear in Section~\ref{sec:ldd}.

A word about the relationship of this note to the work of Miller, Peng,
and Xu~\cite{MPX13}: in their algorithm each vertex $v \in V$ first
picks a random value $X_v \sim \text{Exp}(\ln n/\Delta)$, and say
$X_{\max} := \max_v X_v$. Their algorithm builds BFS trees at unit rate
from a set of terminals, where we start off with the terminal set being
empty, and each vertex $v$ enters the terminal set (and hence starts
building its BFS tree) at time $X_{\max} - X_v$. Each vertex is assigned
to the first BFS tree it belongs to. We can think of this as building
\emph{additively weighted Voronoi diagrams}. In contrast, we choose a
set of terminals that are fixed over time, but our BFS trees grow at
random rates --- this is more akin to \emph{multiplicatively weighted}
Voronoi diagrams. Their algorithm is parallelizable, and also gives
strong-diameter decompositions, whereas we only give weak-diameter
decompositions.
On the other hand, our algorithm is naturally scale-free and hence lends
itself more naturally to terminal partitioning and $0$-extension,
whereas the~\cite{MPX13} algorithm is scale-based and more natural for
low-diameter decompositions.

\section{The Terminal Partitioning Problem}
\label{sec:zero-extension}

\textbf{Input:} given a metric $(V,d)$ and terminals $T \sse V$, where
$n := |V|$ and $k := |T|$.

\textbf{Output:} a (random) map $f: V \to T$ such that
\begin{OneLiners}
\item[(i)] (retraction) $f(t) = t$ for all $t \in T$,
\item [(ii)](separation) for all $u,v \in V$, we have
  \begin{gather}
    \Pr[ f(u) \neq f(v) ] \leq \alpha \cdot \frac{d(u,v)}{\min(A_u, A_v)}, \label{eq:1}
  \end{gather}
  where $A_u := d(u,T)$ is the distance from $u$ to its closest terminal
  in $T$. 
\end{OneLiners}
Such a (random) map $f$ is called a \emph{terminal partitioning with
  stretch} $\alpha$.  There is an optional property that will be useful:
\begin{OneLiners}
\item[(iii)] Let $\B(x, r) := \{ y \in V \mid d(x,y) \leq r\}$ be the
  radius-$r$ ball around $x$ in the metric $(V,d)$. For $c > 0$, the map
  $f$ is \emph{$c$-proximate} if for all $u \in V$, 
  \[ \Pr [ f(u) \in \B(u, c\cdot A_u) ] = 1. \]
\end{OneLiners}
Note that if a mapping satisfies the proximity property~(iii), it also
satisfies the retraction property~(i), simply because each terminal $t$
has $A_t = 0$, hence $f(t) \in \B(t, 0) \implies f(t) = t$.

An $\alpha$-stretch algorithm for terminal partitioning immediately
implies an $\alpha$-approximation for the $0$-extension problem (which
we do not define here); for details, see the original paper of Calinescu
et al.~\cite{CKR01-zero}.

\section{An Algorithm for Terminal Partitioning}
\label{sec:algo-zeroex}

We now give the algorithm for terminal partitioning. We first define the
truncated exponential distribution. Given parameters $\lambda$ and
$\gamma > 0$, the distribution $\text{TExp}(\lambda,\gamma)$ is simply
the exponential distribution $\text{Exp}(\lambda)$ conditioned on being
at most $\gamma$. Formally it is supported on $[0,\gamma]$ and has
density at $x \in [0,\gamma]$ equal to $p(x) = Z(\lambda,\gamma) \cdot
\lambda \exp(-\lambda x)$. Here $Z(\lambda,\gamma) = (1-\exp(-\lambda
\gamma))^{-1}$ is a normalization term. Some useful properties of this
distribution, which we use in the following analysis, can be found in
Section~\ref{sec:trunc-exp}. 

\subsection{The Random-Rates Algorithm}
\label{sec:trunc-exp}

Let $K\geq 3$ be a parameter such that for every vertex $x$, $|T \cap
B(x,2A_x)| \leq K$. Clearly $K \leq \max(3,|T|) = \max(3,k)$. 

\textbf{Algorithm Random-Rates}
\begin{OneLiners}
\item[(a)] For each terminal $t$, independently set $\nu_t \sim
  \text{TExp}(\ln K,1)$. 
\item[(b)] For each terminal $t$, set its ``rate'' $\rho_t \gets 1+\nu_t$.
\item[(c)] Imagine growing ``Voronoi'' regions at rate $\rho_t$ around each
  terminal $t$ to capture vertices. Formally, define the retraction $f$ as
  \begin{gather}
    f(x) = \argmin_{t \in T} \left\{ \frac{d(x,t)}{\rho_t}  \right\} \label{eq:2}
  \end{gather}
  We break ties arbitrarily. 
\end{OneLiners}

The main theorem of this section is the following:
\begin{theorem}
  \label{th:main}
  The random map $f$ defined by Algortithm Random-Rates is a terminal
  partitioning with stretch $\alpha = O(\log K)$, and is $2$-proximate.
\end{theorem}
The proof appears in the next section. Moreover, the paper \cite{FHRT03}
shows that for any map that satisfies the $2$-proximity condition, the
stretch of $O(\log k)$ is best possible. In Section~\ref{sec:ldd} we
will see another proof of this optimality.

\subsection{Proof of Theorem~\ref{th:main}}
\label{sec:proof}


It is easy to see the 2-proximity. Indeed, by definition,
each $\rho_t \in [1,2]$. If $t_x$ is the terminal closest to $x$, then
the definition of $f$ ensures that 
\[ \frac{d(x,f(x))}{2} \leq \frac{d(x,f(x))}{\rho_{f(x)}} \leq
\frac{d(x,t_x)}{\rho_{t_x}} \leq d(x,t_x). \] It follows that $d(x,f(x))
\leq 2A_u$, which proves the map $f$ is 2-proximate.

To prove the stretch bound, we will show a stronger {\em padding}
property. For any $u \in V$, and any $r\geq 0$, we say that the ball
$\B(u,r)$ {\em is cut} (by the mapping $f$) if there exists $v \in
\B(u,r)$ such that $f(u) \neq f(v)$. We say that a terminal $t$ {\em
  captures} $u$ if $f(u)=t$, and that \emph{$t$ cuts $\B(u,r)$} if $t$
captures $u$ and $\B(u,r)$ is cut.

\begin{lemma}
  \label{lem:pad}                 
  For any $u \in V$ and any radius $r \leq A_u/4$,
  \begin{align}
    \Pr[\B(u,r) \mbox{ is cut} ] \leq O(\log K) \cdot
    \frac{r}{A_u}.\label{eq:padding}
  \end{align}
\end{lemma}

\begin{proof}
  Fix a terminal $t^\star$. We first upper bound $\Pr[\B(u,r) \mbox{ is
    cut by } t^\star]$. Note that by the $2$-proximity condition, it
  suffices to consider $t^\star$ such that $d(u,t^\star) \in
  [A_u,2A_u]$. Condition on the rates $\hatr_t$ for all other terminals
  $t \neq t^\star$, and define the ``critical threshold'' for $x \in V$
  to be 
  \begin{gather}
    \rhothresh_{t^\star}(x) := d(x,t^\star) \cdot \argmax_{t \in T: t
      \neq t^\star}\left\{ \frac{\hatr_t}{d(x,t)}\right\} \label{eq:3}
  \end{gather}
  for all $x \in V$. Note that if $\rho_{t^\star} >
  \rhothresh_{t^\star}(x)$, then
  $f(x)=t^\star$. 
  We first prove a simple lemma.
  
  \begin{lemma}
    \label{lem:threshlip}
    Let $v \in B(u,r)$ for $r \leq A_u/4$, and let $t$ be such that
    $d(u,t)\leq 2A_u$. Then
    \begin{gather}
      \rhothresh_{t}(v) - \rhothresh_{t}(u) \leq \frac{12\,r}{A_u}. \label{eq:6}
    \end{gather}
  \end{lemma}
  \begin{proof}
    First observe that for any $t'$,
    \begin{align*}
      \frac{d(v,t)}{d(v,t')} - \frac{d(u,t)}{d(u,t')} & \leq
      \frac{d(u,t)+r}{d(u,t')-r} - \frac{d(u,t)}{d(u,t')}\\ 
      &\leq \frac{(d(u,t)+r)(1+\frac{2r}{d(u,t')})}{d(u,t')} -
      \frac{d(u,t)}{d(u,t')}\\ 
      &\leq \frac{r + (d(u,t)+r)(\frac{2r}{A_u})}{A_u}\\
      &\leq \frac{r + (\frac{5\,d(u,t)}{4})(\frac{2r}{A_u})}{A_u}\\
      &\leq \frac{r+5r}{A_u}
    \end{align*}
    Thus $\frac{d(v,t) \cdot \hatr_{t'}}{d(v,t')} -
    \frac{d(u,t) \cdot \hatr_{t'}}{d(u,t')} \leq \frac{12r}{A_u}$. The claim
    follows by definition of $\rhothresh$ and Lipschitz-ness of $\max$. 
  \end{proof}
 
  The rest of the proof is relatively simple: when the threshold is far
  from $\gamma$, the truncation has little effect, and the
  memorylessness property of the exponential suffices to show that the
  probability of cutting $\B(u,r)$, conditioned on capturing $u$ is
  small for $t^{\star}$. When the threshold is closer to $\gamma$, this
  conditional probability can be large. However, for such large
  thresholds, the unconditional probability is small enough that we can
  afford to add these probabilities over the $K$ terminals. We formalize
  this next.

  Let $\delta := 12r/A_u$ be the upper bound in~(\ref{eq:6}), and let
  $\lambda := \ln K$, the parameter for the truncated exponential. It
  follows that if $\rho_{t^\star} \geq \rhothresh_{t^\star}(u)+ \delta$,
  then $t^\star$ captures all of $\B(u,r)$. Recall that the definition
  of $t^\star$ cutting $\B(u,r)$ is that $t^\star$ must capture $u$ but
  not all of $\B(u,r)$. Hence,
  \[ \Pr\big[ t^\star \text{ cuts } \B(u,r) \big] \leq \Pr\big[ \rho_{t^\star} \in
  [\rhothresh_{t^\star}(u),\rhothresh_{t^\star}(u)+ \delta) \big] .\]
 
  Observe that if $a \leq 1 - \frac{1}{\lambda}$, then $e^{-\lambda a} -
  e^{-\lambda} = e^{-\lambda a} ( 1 - e^{\lambda(a-1)}) \geq
  \frac{e^{-\lambda a}}{2}$.  Thus if $\rhothresh_{t^\star}(u) \leq 2 -
  \frac{1}{\lambda}$, then recall that $\rho_{t^\star} - 1$ is a
  truncated exponential, and use Proposition~\ref{prop:texpcond}(c) to
  get 
  \begin{align*}
    & \Pr\bigg[t^\star \mbox{ cuts } \B(u,r) \biggm\vert (t^\star \mbox{
      captures } u) \land (\rhothresh_{t^\star}(u)  \leq 2 -
    \frac{1}{\lambda}) \bigg] \\
    &\leq \Pr\bigg[ \rho_{t^\star} \leq \rhothresh_{t^\star}(u)+ \delta
    \biggm\vert (\rho_{t^\star} \geq \rhothresh_{t^\star}(u)) \land
    (\rhothresh_{t^\star}(u) \leq 2 -
    \frac{1}{\lambda}) \bigg] \\
    &\leq \delta\, \lambda
    \cdot \frac{\exp(-\lambda\, \rhothresh_{t^\star}(u)
      )}{\exp(-\lambda\, \rhothresh_{t^\star}(u)) - \exp(-\lambda)}
    \leq 2\delta\lambda. 
  \end{align*}
  On the other hand, if $\rhothresh_{t^\star}(u) > 2-\frac{1}{\lambda}$,
  then by Proposition~\ref{prop:texpcond}(b), 
  \begin{align*}
    \Pr\bigg[t^\star \mbox{ cuts } \B(u,r) \biggm\vert
    \rhothresh_{t^\star}(u) > 2 - \frac{1}{\lambda} \bigg] 
    &=
    \Pr\bigg[ \rho_{t^\star} \in [\rhothresh_{t^\star}(u),
    \rhothresh_{t^\star}(u)+ \delta)
    \biggm\vert 
    \rhothresh_{t^\star}(u) > 2 - \frac{1}{\lambda} \bigg] 
    \\ 
    &\leq 2\delta\lambda e^{-\lambda
    (1-1/\lambda)} = 2\delta\lambda \, e^{1 - \lambda} \leq 2e\delta\lambda/K.
  \end{align*}
  It follows that
  \begin{align*}
    \Pr[t^\star \mbox{ cuts } \B(u,r)] &\leq \Pr[t^\star \mbox{ captures
    } u] \cdot 2\delta \lambda + 2e\delta\lambda/K
  \end{align*}
  Since there are $K$ possible terminals that can capture $u$, and
  exactly one captures $u$, it follows that
  \begin{align*}
    \Pr[\B(u,r) \mbox{ gets cut}] &\leq \bigg(
    \sum_{t^\star} \Pr[t^\star \mbox{ captures
    } u] \bigg) \cdot 2\delta \lambda  +
    K \cdot 2e\delta\lambda/K\\ 
    &\leq 2(1+e)\delta\lambda.
  \end{align*}
  Since $\delta = O(r/A_u)$, and $\lambda = \ln K$, the claim follows.
\end{proof}

Finally, to show that the padding property of Lemma~\ref{lem:pad}
implies the separation probability~(\ref{eq:1}) is standard: we give it
here for completeness. If $d(u,v) \geq A_u/4$, then $O(\log K) \cdot
\frac{d(u,v)}{A_u} \geq 1$ for a large enough constant in the big-Oh,
so~(\ref{eq:1}) is trivially satisfied. Else, $v \in \B(u,r^\star)$ for $r^\star =
d(u,v) \leq A_u/4$, and $\B(u,r^\star)$ not being cut implies that $u,v$ are not
separated; by Lemma~\ref{lem:pad} this happens with probability 
\[ O(\log K) \cdot \frac{r^\star}{A_u} = O(\log K) \cdot \frac{d(u,v)}{A_u}
\leq O(\log K) \cdot \frac{d(u,v)}{\min(A_u,A_v)}. \]
This completes the proof of Theorem~\ref{th:main}.

\section{An Algorithm for Low-Diameter Decompositions}
\label{sec:ldd}

We can get an algorithm for low-diameter decompositions (LDDs) using a
similar random rates idea. Recall that in the LDD problem, we are given a
metric $(V,d)$ and parameter $\Delta$, we want a random partition $V_1,
V_2, \ldots, V_q$ of the point set $V$ such that:
\begin{OneLiners}
\item[(i)] The clusters have diameter at most $\Delta$; i.e., $\max_i
  \max_{x,y \in V_i} d(x,y) \leq \Delta$, and
\item[(ii)] The probability 
  \begin{gather}
    \Pr[ x, y \text{ not in same cluster } ] \leq \beta \cdot \frac{
      d(x,y) }{\Delta}.
  \end{gather}
\end{OneLiners}

Recall that an $\eps$-net of a metric $(V,d)$ is a set $N \sse V$ such
that (a) for all $v \in V$, the distance to the nearest net point is at
most $\eps$ (i.e., $d(v,N) \leq \eps$), and (b) two net points are
$\eps$ apart (i.e., $d(t_1, t_2) \geq \eps$ for $t_1, t_2 \in N$ such
that $t_1 \neq t_2$). A greedy algorithm gives us such a net;
near-linear time algorithms are also known to find nets~\cite{HPM06}.

Our LDD procedure is the following simple reduction: 
\begin{quote}
  \textbf{Algorithm Random-Rates-LDD:} Let $T$ be a $\Delta/10$-net
  of $(V,d)$. Use a $2$-proximate terminal partitioning algorithm to
  define the clusters in the natural way: the vertices that map to the
  same terminal in $T$ are in the same cluster.
\end{quote}

\begin{lemma}
  \label{lem:ldd}
  A $2$-proximate terminal partitioning $f$ with stretch $\alpha$ gives
  us a $\Delta$-LDD with $\beta = O(\alpha)$.
\end{lemma}

\begin{proof}
  Consider $x,y$ such that $d(x,y) > \Delta$, we claim that $f(x) \neq
  f(y)$. Indeed, since we found a $\Delta/10$-net, the closest terminal
  to each node is at distance at most $\Delta/10$ from it. By the
  proximity property, each node is assigned to a terminal at distance at
  most $\Delta/5$ from it, and since $d(x,y) > \Delta$, we must have
  $f(x) \neq f(y)$ by the triangle inequality. Hence we have the
  low-diameter property.

  Now for the probability of separation for some pair $x,y$. For $x,y$
  which are ``far apart'', say, $d(x,y) > \Delta/100$, the probability
  that $x,y$ are separated is trivially at most $1$, which is at most
  $100 \cdot d(x,y)/\Delta$, so $\beta \geq 100$ suffices for them.

  So assume $d(x,y) \leq \Delta/100$. Let $t_x, t_y$ be the closest
  terminals to $x,y$ respectively, and so $A_x = d(x,t_x)$ and $A_y =
  d(y, t_y)$. There are two cases:
  \begin{itemize}
  \item Both $A_x, A_y \geq \Delta/100$. Then by~(\ref{eq:2}), we have the
    probability of $x,y$ separated (or equivalently $f(x) \neq f(y)$) is
    at most
    \[ \alpha \cdot \frac{d(x,y)}{\min(A_x,A_y)} \leq 100\,\alpha \cdot
    \frac{d(x,y)}{\Delta}. \]
  \item At least one of $A_x, A_y \leq \Delta/100$, say $A_x \leq
    A_y$. Then $A_y \leq A_x + d(x,y) \leq \Delta/100 + \Delta/100 =
    \Delta/50$.   Since we also have $d(t_x,t_y) \leq d(x,t_x) + d(x,y) +
    d(y,t_y) = A_x + A_y + d(x,y) \leq \Delta/25$. By the packing
    property of a $\Delta/10$-net, we know that if $t_x
    \neq t_y$ then $d(t_x, t_y) \geq \Delta/10$, which implies that $t_x
    = t_y$. 

    Moreover, consider any other terminal $t$ within $B(x,2A_x) \cup
    B(y,2A_y)$, then $d(t_x, t) \leq 3A_x$ or $d(t_x, t) \leq A_x +
    d(x,y) + 2A_y$. 
    In either case, this would mean $d(t_x, t) \leq 6\Delta/100$, and
    hence again $t_x = t$. In other words, the only terminal within
    distance $2A_x$ of $x$ (and within $2A_y$ of $y$) is $t_x = t_y$.
    Now by the proximity condition, $f(x) = f(y)$ with probability~$1$.
  \end{itemize}
  This shows that the LDD procedure above satisfies $\beta \leq 100 \alpha$.
\end{proof}

Since the size of the net is at most $n$, this implies $\beta = O(\log
n)$. Moreover, recall that a metric has {\em doubling dimension} $\dim$
if for all $u \in V$ and $r \geq 0$, any set of diameter $2r$ can be
covered by $2^{\dim}$ sets of diameter at most $r$. It is a standard
fact that for metrics of doubling dimension $\dim$, any net $T$ has the
property that for every $u \in V$, $|\B(u,2A_u) \cap T| \leq
2^{O(\dim)}$. Thus $K$ is $2^{O(\dim)}$, and we get an LDD with
parameter $\beta = O(\dim)$, matching known results~\cite{GKL03}. We
summarize these results below.

\begin{corollary}
  \label{cor:ldd}
  Algorithm Random-Rates-LDD, using the random map $f$ from
  Section~\ref{sec:algo-zeroex}, has parameter $\beta = O(\log
  n)$. Moreover, for metrics of constant doubling dimension, the
  parameter $\beta = O(1)$.
\end{corollary}

It is known that for LDDs on general metrics, $\beta = \Omega(\log n)$
is best possible, e.g., for large girth expanders (see,
e.g.,~\cite{Bar96}). The above reduction gives another proof that for
$O(1)$-proximate terminal partitionings, we cannot achieve $\alpha =
o(\log k)$.

\section{Properties of the Truncated Exponential Distribution}
\label{sec:trunc-exp}

Here are some properties of the truncated exponential that were useful
in our analysis.
\begin{proposition}
\label{prop:texpcond}
Let $\nu \sim \text{TExp}(\lambda,\gamma)$, and $a,b>0$ be such that $(a+b) \leq \gamma$. Suppose $\gamma > 1/\lambda$. Then
\begin{OneLiners}
\item[(a)] $Z(\lambda,\gamma) = (1-\exp(-\lambda \gamma))^{-1} \leq 2$.
\item[(b)] $\Pr[\nu \in (a,a+b)] \leq 2 \exp(-\lambda a)(1-\exp(-b\lambda)) \leq 2 b \lambda \exp(-\lambda a)$.
\item[(c)] $\Pr[\nu \leq (a+b) \mid \nu \geq a] = \frac{\exp(-\lambda a) - \exp(-\lambda(a+b))}{\exp(-\lambda a) - \exp(-\lambda \gamma)} \leq b\lambda\cdot \frac{\exp(-\lambda a)}{\exp(-\lambda a) - \exp(-\lambda \gamma)}$.
\end{OneLiners}
\end{proposition}

\begin{proof}
  Part~(a) follows from $\gamma \lambda > 1$ and hence $Z(\lambda,
  \gamma) = (1-\exp(-\lambda \gamma))^{-1} \leq (1 - \exp(1))^{-1} =
  \frac{e}{e-1} \leq 2$. For part~(b), we have
  \begin{align*}
    \Pr[\nu \in (a,a+b)] &= Z(\lambda, \gamma) \cdot \lambda \cdot \int_{x
      = a}^{a+b} e^{-\lambda x} \\
    &= Z(\lambda, \gamma) \cdot ( e^{-\lambda a} - e^{-\lambda (a+b)})
    \\
    &\leq 2 e^{-\lambda a}( 1
    - e^{- \lambda b}) \leq 2b \lambda e^{-\lambda a}.
  \end{align*}
  The last step uses part~(a), and that $1 + y \leq e^y$ for all $y\in \R$.
  For part~(c), we use similar calculations.
\end{proof}

\subsection*{Acknowledgments}

This work was done when A.~Gupta was visiting Microsoft Research SVC in
2006; he thanks them for their hospitality. We also thank T.-H.\ Hubert
Chan and Satish Rao for useful discussions.

{\small
\bibliographystyle{alpha}
\bibliography{zeroext}
}

\end{document}